\documentclass[12pt]{article}
\usepackage[english]{babel}
\usepackage{amsmath,amssymb,amsfonts,latexsym}
\usepackage{color}
\usepackage{graphicx}
\newtheorem{theorem}{Theorem}

\newtheorem{lemma}{Lemma}

\newtheorem{example}{Example}
\newtheorem{remark}{Remark}
\newenvironment{proof}{\noindent {\bf Proof.}}{\hfill$\Box$}

\begin{document}

\title{Nonlinear first order PDEs reducible to autonomous form polynomially homogeneous in the derivatives}

\author{M.~Gorgone and F.~Oliveri\\
\ \\
{\footnotesize MIFT Department, University of Messina}\\
{\footnotesize Viale F. Stagno d'Alcontres 31, 98166 Messina, Italy}\\
{\footnotesize mgorgone@unime.it; foliveri@unime.it}
}

\date{Published in \textit{J. Geom. Phys.} \textbf{113}, 53--64 (2017).}

\maketitle

\begin{abstract}
It is proved a theorem providing necessary and sufficient conditions enabling one to map a nonlinear system 
of first order partial differential equations, polynomial in the derivatives, to an equivalent autonomous first 
order system polynomially homogeneous in the derivatives. The result is intimately related to the symmetry 
properties of the source system, and the proof, involving the use of the canonical variables associated to the 
admitted Lie point symmetries,  is constructive. First order Monge-Amp\`ere systems, either with constant coefficients or with coefficients depending on the field variables, where the theorem can be successfully applied, are considered.
\end{abstract}


\section{Introduction}
\label{sec:introduction}
Lie group analysis 
\cite{Ovsiannikov,Ibragimov, Olver, Olver1995,Baumann, BlumanAnco,Meleshko2005, BlumanCheviakovAnco2009}
provides a unified and elegant algorithmic framework to a deep understanding and fruitful handling 
of differential equations. It is known that Lie point symmetries admitted by ordinary differential equations
allow for their order lowering and possibly reducing them to quadrature, whereas in the case of partial differential equations they can be used for the determination of special 
(invariant) solutions of initial and boundary value problems. Also, the Lie symmetries are important ingredients
in the derivation of conserved quantities, or in the construction of relations between different differential 
equations that turn out to be equivalent  \cite{BlumanCheviakovAnco2009,Bluman:article,BlumanKumei1989,
DonatoOliveri1994, DonatoOliveri1995, DonatoOliveri1996, CurroOliveri2008, 
Oliveri2010,Oliveri2012}.
Lie point symmetries of differential equations, in fact,  can be used to construct a mapping from a given
(source) system of differential equations to another (target) suitable system; if we
consider one-to-one (invertible) point mappings, then a one-to-one
correspondence between Lie point symmetries admitted by the source and target system of differential
equations arises. In other words, the Lie algebra of infinitesimal operators of the target system of differential
equations has to be isomorphic to the Lie algebra of infinitesimal operators of the source system of
differential equations. This property has been used to give necessary and sufficient conditions for reducing a system of partial differential equations to autonomous form \cite{DonatoOliveri1995,DonatoOliveri1996}, a system of first
order nonlinear partial differential equations to linear form \cite{Bluman:article,BlumanKumei1989,DonatoOliveri1994,DonatoOliveri1996}, nonautonomous and/or nonhomogeneous quasilinear  systems of partial differential equations to autonomous and homogeneous form \cite{CurroOliveri2008, Oliveri2012}. In particular, in \cite{Oliveri2012}, 
it has been proved a theorem providing the necessary and sufficient conditions in order to map a general first order quasilinear system of partial differential equations, say
\begin{equation}
\label{quasilinear}
\sum_{i=1}^n A^i(\mathbf{x},\mathbf{u})\frac{\partial\mathbf{u}}{\partial x_i}=\mathbf{B}(\mathbf{x},\mathbf{u}),
\end{equation}
where $\mathbf{x}\in\mathbb{R}^n$, $\mathbf{u}\in\mathbb{R}^m$, $A^i$ are $m\times m$ matrices 
with entries depending at most on $\mathbf{x}$ and $\mathbf{u}$, and the source term 
$\mathbf{B}\in \mathbb{R}^m$ depends at most on $\mathbf{x}$ and $\mathbf{u}$ too,
into a first order quasilinear homogeneous and autonomous system. This reduction, when it is possible, 
is performed by an invertible point transformation like
\begin{equation}
\mathbf{z}=\mathbf{Z}(\mathbf{x}), \qquad \mathbf{w}=\mathbf{W}(\mathbf{x},\mathbf{u}),
\end{equation}
which preserves the quasilinear structure of the system, and whose construction is algorithmically suggested 
by the Lie symmetries admitted by (\ref{quasilinear}).

In this paper, we consider a general nonlinear system of first order partial differential equations involving the derivatives of the unknown variables in polynomial (of degree greater than 1) form, and establish 
a theorem giving necessary and sufficient conditions in order to map it to an autonomous system which is polynomially homogeneous in the derivatives. 

In some relevant situations, \emph{e.g.}, Monge--Amp\`ere systems, 
the target system results to be quasilinear, but there are cases where the system we obtain is polynomially homogeneous in the derivatives but not quasilinear. This means that the conditions of the theorem are only necessary for the reduction of a \emph{nonlinear} first order system to autonomous and homogeneous \emph{quasilinear} form \cite{GorgoneOliveri2016}. 

The main difference of the theorem presented in this paper with the similar one proved in \cite{Oliveri2012} (concerned with the transformation of a general first order quasilinear system of partial differential equations into a first order quasilinear homogeneous and autonomous
system) consists in the possibility of admitting now an
 invertible point transformation like
\begin{equation}
\mathbf{z}=\mathbf{Z}(\mathbf{x},\mathbf{u}), \qquad \mathbf{w}=\mathbf{W}(\mathbf{x},\mathbf{u}),
\end{equation}
\emph{i.e.}, a mapping where the new independent variables $\mathbf{z}$ are allowed to depend also on the old dependent ones.

The plan of the paper is the following. In Section~\ref{main}, 
the theorem giving necessary and sufficient conditions for the existence of an
invertible mapping linking a nonlinear system of first order partial differential equations which is polynomial in the derivatives to an autonomous system  polynomially homogeneous in the derivatives is proved.
In Section~\ref{applications}, the theorem is applied to various general first order Monge--Amp\`ere systems. Finally,  Section~\ref{conclusions}  contains some concluding remarks.

\section{Main result}
\label{main}
Let us consider a general system of first order partial differential equations
\begin{equation}
\label{generalsystem}
\Delta\left(\mathbf{x},\mathbf{u},\mathbf{u}^{(1)}\right)=0,
\end{equation}
where $\mathbf{x}\in\mathbb{R}^n$, $\mathbf{u}\in\mathbb{R}^m$ and $\mathbf{u}^{(1)}\in\mathbb{R}^{mn}$
are the independent variables, the dependent variables, and the first order partial derivatives, respectively.
In particular, in the following we consider systems (\ref{generalsystem}) composed by equations which are polynomial in the derivatives, with coefficients depending at most on $\mathbf{x}$ and $\mathbf{u}$, \emph{i.e.} systems made by equations of the form
\begin{equation}
\label{polynomial}
\sum_{|\boldsymbol\alpha|,|\mathbf{j}|=1}^{N_s} A^s_{\boldsymbol\alpha \mathbf{j}}(\mathbf{x},\mathbf{u})
\prod_{k=1}^{|\boldsymbol\alpha|}\frac{\partial u_{\alpha_k}}{\partial x_{j_k}}+B^s(\mathbf{x},\mathbf{u})=0,\qquad s=1,\ldots,m,
\end{equation}
where $\boldsymbol\alpha$ is the multi-index $(\alpha_1,\ldots,\alpha_r)$, $\mathbf{j}$ the multi-index
$(j_1,\ldots,j_r)$,  $\alpha_k=1,\ldots,m$, $j_k=1,\ldots,n$, $N_s$ are integers,
and $A^s_{\boldsymbol\alpha \mathbf{j}}(\mathbf{x},\mathbf{u})$, $B^s(\mathbf{x},\mathbf{u})$ smooth functions
of their arguments.
 
The aim is to determine necessary and sufficient conditions for the construction of an invertible point transformation
\begin{equation}
\label{map}
\mathbf{z}=\mathbf{Z}\left(\mathbf{x},\mathbf{u}\right), \qquad
\mathbf{w}=\mathbf{W}\left(\mathbf{x},\mathbf{u}\right), 
\end{equation}
mapping (\ref{polynomial}) into an equivalent autonomous system 
which is  homogeneous polynomial in the derivatives $\mathbf{w}^{(1)}$, \emph{i.e.}, made by equations of the form
\begin{equation}
\label{homopolynomial}
\sum_{|\boldsymbol\alpha|,|\mathbf{j}|=\overline{N}_s} \widetilde{A}^s_{\boldsymbol\alpha\mathbf{j}}(\mathbf{w})
\prod_{k=1}^{\overline{N}_s}\frac{\partial w_{\alpha_k}}{\partial z_{j_k}}=0, \qquad s=1,\ldots,m,
\end{equation}
for some integers $\overline{N}_s$; of course, it may occur that the target system turns out to be linear in the derivatives, \emph{i.e.},  $\overline{N}_s=1$ $(s=1,\ldots,m)$, whereupon we
have an autonomous and homogeneous quasilinear system.

The following lemma, guarantees  that an invertible point transformation like (\ref{map}) preserves the polynomial
structure in the derivatives. 

\begin{lemma}
Given a first order system of partial differential equations like (\ref{polynomial}) which is polynomial in the derivatives, then an invertible point transformation like (\ref{map})
produces a first order system which is still polynomial in the derivatives.
\end{lemma}
\begin{proof}
Straightforward, by using the chain rule.
\end{proof}

\begin{theorem}\label{th_1}
The nonlinear first order  system of partial differential equations polynomial in the derivatives
\begin{equation}
\label{source}
\sum_{|\boldsymbol\alpha|,|\mathbf{j}|=1}^{N_s} A^s_{\boldsymbol\alpha \mathbf{j}}(\mathbf{x},\mathbf{u})
\prod_{k=1}^{|\boldsymbol\alpha|}\frac{\partial u_{\alpha_k}}{\partial x_{j_k}}+B^s(\mathbf{x},\mathbf{u})=0,\qquad s=1,\ldots,m,
\end{equation}
is mapped by an invertible point transformation, say
\begin{equation}
\label{source2target}
\mathbf{z}=\mathbf{Z}(\mathbf{x},\mathbf{u}), \qquad \mathbf{w}=\mathbf{W}(\mathbf{x},\mathbf{u}),
\end{equation}
to the equivalent nonlinear first order autonomous system having homogeneous polynomial form, say
\begin{equation}
\label{target}
\sum_{|\boldsymbol\alpha|,|\mathbf{j}|=\overline{N}_s} \widetilde{A}^s_{\boldsymbol\alpha\mathbf{j}}(\mathbf{w})
\prod_{k=1}^{\overline{N}_s}\frac{\partial w_{\alpha_k}}{\partial z_{j_k}}=0, \qquad s=1,\ldots,m,
\end{equation}
for some integers $\overline{N}_s$,
if and only if there exists an $(n+1)$--dimensional subalgebra of the Lie algebra of point symmetries, admitted by system (\ref{source}), spanned by the vector fields
\begin{equation}
\Xi_i=\sum_{i=1}^n\xi_i^j(\mathbf{x},\mathbf{u})\frac{\partial}{\partial x_i}+\sum_{\alpha=1}^m\eta_i^\alpha(\mathbf{x},\mathbf{u})\frac{\partial}{\partial u_\alpha}, \qquad i=1,\ldots,n+1,
\end{equation}
generating a distribution of rank $(n+1)$, and such that
\begin{equation}
\label{condalgebra}
\begin{aligned}
&\left[\Xi_i,\Xi_j\right]=0, \qquad &&i=1,\ldots,n-1,\qquad i<j\le n,\\
&\left[\Xi_i,\Xi_{n+1}\right]=\Xi_i, \qquad &&i=1,\ldots, n.
\end{aligned} 
\end{equation}
Moreover, the variables $\mathbf{w}$, which by construction are invariants of $\Xi_1$, $\ldots$, $\Xi_n$, 
have to be invariant with respect to $\Xi_{n+1}$ too.
\end{theorem}
\begin{proof}
Suppose the conditions of Theorem \ref{th_1} are satisfied, and so the system (\ref{source}) admits an 
$(n+1)$--dimensional algebra as subalgebra of the algebra of its Lie point symmetries generating a distribution of rank $(n+1)$ and verifying the structure conditions (\ref{condalgebra}).
Let us introduce a set of canonical variables for the vector field $\Xi_1$, say
\[
y^1_i\quad (i=1,\ldots,n), \qquad
v^1_\alpha\quad (\alpha=1,\ldots,m),
\]
such that
\[
\Xi _{1} y^1_1=1, \qquad \Xi _{1} y^1_{i_1}=0,\qquad \Xi _{1} v^1_\alpha=0
\]
($i_1=2,\ldots,n$); as a consequence, $\Xi_{1}$  takes the form
\[
\Xi_{ 1}={\frac{\partial}{\partial y^1_{ 1}}},
\]
\emph{i.e.}, it corresponds to a translation in the variable $y^1_1$.

Since $[\Xi_{1},\Xi_{2}]=0$,  it is
\begin{equation}
\Xi_{1} (\Xi _{2} y^1_i)= \Xi_{ 2} (\Xi_{1} y^1_i )=0, \qquad
\Xi_{1} (\Xi _{2} v^{1}_{\alpha} )=\Xi_{2} (\Xi _{1} v^{1}_{\alpha})=0,
\end{equation}
for $i=1,\ldots,n$ and $\alpha=1,\ldots,m$.
Thus, the infinitesimals of $\Xi_{2}$, represented in terms of the canonical variables of $\Xi_1$,  
will depend upon the invariants of $\Xi_1$ only, \emph{i.e.}, $\Xi_2$ writes as
\begin{equation}
\Xi_{ 2} =\sum_{i=1}^{n}\Theta ^{2}_{i} (y^{1}_{j_{1}}, v^{1}_{\beta}) {\frac{\partial}{\partial y^1_{i}}} 
+\sum_{\alpha=1}^{m}
\Lambda^{2}_{ \alpha} (y^{1}_{j_{ 1}}, v^{1}_{\beta}) {\frac{\partial}{\partial v^{1}_{ \alpha} }}
\end{equation}
($j_{ 1} =2,\ldots,n$, $\beta=1,\ldots,m$).

If $\Theta ^{2}_{1} \neq 0$ we need to replace $y^1_1$ with
\[
y^1_{ 1} +\varphi^1_{1} (y^{1}_{ j_{1} }, v^{1}_{\beta}),
\]
where the function $\varphi^1_{ 1}$  satisfies
\[
\Theta^{2}_{1}(y^{1}_{j_{1}}, v^{1}_{\beta})+\sum_{i_{1} =2}^{n} \Theta^{2}_{i_{1}}(y^{1}_{j_{1}},v^1_\beta)
{\frac{\partial\varphi^1 _{1}}{\partial y^{1}_{i_{1}}}}+\sum_{\alpha =1}^{m} \Lambda^{2}_{\alpha}(y^{1}_{j_{1}},v^1_\beta)
{\frac{\partial\varphi^1 _{1}}{\partial v^{1}_{\alpha}}} =0.
\]

That enables us to write $\Xi_{1}$ and $\Xi_{2}$  as follows
\begin{equation}
\Xi_{1} ={\frac{\partial }{\partial y^{1}_{ 1} }} ,\qquad
\Xi_{2}=\sum_{i_{1}=2}^{n}\Theta^{2}_{ i_{1}}(y^{1}_{j_{1}},v^1_\beta)
{\frac{\partial }{\partial y^{1}_{i_{1}}}} + \sum_{\alpha=1} ^{m}
\Lambda ^{2}_{ \alpha} (y^{1}_{j_{1}}, v^{1}_{\beta}) {\frac{\partial}{\partial v^{1}_{\alpha} }},
\end{equation}
where $j_{1} =2,\ldots,n$.

Introducing the canonical variables
\[
y^2_1=y^1_1,\qquad y^2_2,\qquad
y^{2}_{i_{2}}\quad (i_{ 2} =3,\ldots,n), \qquad
v^{2}_{\alpha}\quad(\alpha=1,\ldots,m),
\]
such that
\[
\Xi _{2} y^2_2=1, \qquad \Xi _{2} y^2_{i_2}=0,\qquad \Xi _{2} v^2_\alpha=0
\]
($i_1=2,\ldots,n$),
 it is obtained
\[
\Xi_2={\frac{\partial}{\partial y^2_2}}.
\]

Continuing inductively for $k=2,\ldots,n-1$, since $\Xi_{k+1}$ commutes with $\Xi_1$, $\ldots$,
$\Xi_k$, in terms of the canonical variables
\[
y^{k}_i\quad (i=1,\ldots, n), \qquad
v^{k}_{\alpha}\quad(\alpha=1,\ldots,m),
\]
we have
\begin{equation}
\begin{aligned}
&\Xi_{1} ={\frac{\partial }{\partial y^{k}_{1}}}, \qquad \Xi_{2}
={\frac{\partial }{\partial y^{k}_{2}}} ,\qquad \dots\dots,\qquad
\Xi_{k} ={\frac{\partial }{\partial y^{k}_{k}}}, \\
&\Xi_{k+1}=\sum_{i=1}^{n}\Theta^{k+1}_{i}(y^{k}_{j_{k}},v^k_\beta)
{\frac{\partial}{\partial y^{k}_{i}}}+\sum_{\alpha=1}^{m}
\Lambda^{k+1}_{\alpha}(y^{k}_{j_{k}}, v^{k}_{\beta} ) {\frac{\partial
}{\partial v^{k}_{ \alpha} }},
\end{aligned}
\end{equation}
where $j_{k}=k+1,\ldots,n$. 
If $\Theta ^{k+1}_{\ell} \neq 0$, for $\ell=1,\ldots,k$, we need to replace the variable
$y^k_\ell$ with
\[
y^{k}_{\ell} +\varphi^k_{\ell} (y^{k}_{ j_{k} },v^k_\beta),
\]
where the function $\varphi^k_{\ell}$  satisfies
\[
\Theta^{k+1}_{\ell}(y^{k}_{j_{k}},v^k_\beta)+\sum_{i_{k} =k+1}^{n} \Theta^{k+1}_{i_{k}}(y^{k}_{j_{k}},v^k_\beta)
{\frac{\partial\varphi^k _{\ell}}{\partial y^{1}_{i_{k}}}}+\sum_{\alpha=1}^m \Lambda^{k+1}_\alpha(y^k_{j_k},v^k_\beta)
\frac{\partial \varphi^k_\ell}{\partial v^{k}_\alpha} =0,
\]
so that $\Xi_{k+1}$ writes as
\begin{equation}
\Xi_{k+1}=\sum_{i_k=k+1}^{n}\Theta^{k+1}_{i_k}(y^{k}_{j_{k}},v^k_\beta)
{\frac{\partial}{\partial y^{k}_{i_k}}}+\sum_{\alpha=1}^{m}
\Lambda^{k+1}_{\alpha}(y^{k}_{j_{k}}, v^{k}_{\beta} ) {\frac{\partial
}{\partial v^{k}_{\alpha} }};
\end{equation}
hence, we may construct the canonical variables
\[
y^{k+1}_1=y^k_1,\;\ldots,\; y^{k+1}_k=y^k_k,\quad
y^{k+1}_{i_{k}}\; (i_{ k} =k+1,\ldots,n), \qquad
v^{k+1}_{\alpha}\;(\alpha=1,\ldots,m),
\]
related to the operator $\Xi_{k+1}$, such that the latter
writes as
\[
\Xi_{k+1}={\frac{\partial}{\partial y^{k+1}_{k+1}}}.
\]

The complete application of the described algorithm
enables us to write each operator $\Xi _{i}$  in the form
\[
\Xi_{ i} ={\frac{\partial }{\partial z_{i}}}, \qquad i=1,\ldots,n,
\]
and the new independent and dependent variables are $z_{i} =y^{n}_{i}$ ($i=1,\ldots, n$), $w_\alpha=v^n_\alpha$ ($\alpha=1,\ldots,m$), respectively.

Therefore, what we have obtained is a variable transformation like
(\ref{source2target}) allowing to write the system (\ref{source}) in
autonomous form.

Finally, since $[\Xi_i,\Xi_{n+1}]=\Xi_i$ ($i=1,\ldots,n$),  it is
\begin{equation}
\begin{aligned}
&\Xi_i(\Xi_{n+1}z_j)=\Xi_{n+1}(\Xi_i z_j)+\Xi_i z_j=\delta_{ij},\\
&\Xi_i(\Xi_{n+1}w_\alpha)=\Xi_{n+1}(\Xi_i w_\alpha)+\Xi_i w_\alpha=0;
\end{aligned}
\end{equation}
where $\delta_{ij}$ is the Kronecker symbol; these relations, together with the
hypothesis that the variables $w_\alpha$ ($\alpha=1,\ldots,m$)
are invariant with respect to $\Xi_{n+1}$, allow the vector field $\Xi_{n+1}$ to gain the  representation
\begin{equation}
\Xi_{n+1}=\sum_{j=1}^n z_j\frac{\partial}{\partial z_j}.
\end{equation}
As a consequence, since the resulting system, written in the variables $\mathbf{z}$ and $\mathbf{w}$, is 
autonomous and polynomial in the derivatives, and is invariant with respect to a uniform scaling of all independent variables, then it
necessarily must be polynomially homogeneous in the derivatives, 
\emph{i.e.}, it has the form (\ref{target}).

The condition that the symmetries generate a distribution of rank $(n+1)$, whence the vector fields spanning the $n$--dimensional Abelian Lie subalgebra generate a distribution of rank $n$,  ensures that we may construct the complete set of the new independent variables $\mathbf{z}$.

Conversely, if the nonautonomous and/or nonhomogeneous 
system (\ref{source}) can be
mapped by the invertible point transformation (\ref{source2target}) to
the autonomous system polynomially homogeneous in the derivatives (\ref{target}), then, since the latter
admits the $n$ vector fields $\displaystyle\frac{\partial}{\partial z_{i}}$,
spanning an $n$--dimensional Abelian Lie algebra, and the vector field
$\displaystyle\sum_{j=1}^n z_j\frac{\partial}{\partial z_{j}}$, then it follows
that also the system (\ref{source}) must admit $(n+1)$ Lie point symmetries
with the requested algebraic structure. 
\end{proof}

\begin{example}
\label{mov}
Let us consider the first order system made by the equations
\begin{equation}
\label{Oliv_sys}
\begin{aligned}
& \frac{\partial u_1}{\partial x_2}-\frac{\partial u_2}{\partial x_1}=0,\\
& \kappa_1 \left(\frac{\partial u_1}{\partial x_2}\right)^4+\left(\kappa_2 \frac{\partial u_1}{\partial x_1}\frac{\partial u_2}{\partial x_2}+\kappa_3 \left(\frac{\partial u_1}{\partial x_2}\right)^2\right)\frac{\partial u_1}{\partial x_1}\frac{\partial u_2}{\partial x_2}
\\
&+\left(\kappa_4 \frac{\partial u_1}{\partial x_1}\frac{\partial u_2}{\partial x_2}+\kappa_5 \left(\frac{\partial u_1}{\partial x_2}\right)^2\right)\frac{\partial u_1}{\partial x_1}+\left(\kappa_6 \frac{\partial u_1}{\partial x_1}\frac{\partial u_2}{\partial x_2}+\kappa_7 \left(\frac{\partial u_1}{\partial x_2}\right)^2\right)\frac{\partial u_1}{\partial x_2}\\
&+\left(\kappa_8 \frac{\partial u_1}{\partial x_1}\frac{\partial u_2}{\partial x_2}
+\kappa_9 \left(\frac{\partial u_1}{\partial x_2}\right)^2\right)\frac{\partial u_2}{\partial x_2}
+\kappa_{10}\left(\frac{\partial u_1}{\partial x_1}\right)^2
+\kappa_{11}\frac{\partial u_1}{\partial x_1}\frac{\partial u_1}{\partial x_2}\\
&+\kappa_{12}\frac{\partial u_1}{\partial x_1}\frac{\partial u_2}{\partial x_2}+\kappa_{13}\left(\frac{\partial u_1}{\partial x_2}\right)^2+\kappa_{14}\frac{\partial u_1}{\partial x_2}\frac{\partial u_2}{\partial x_2}+\kappa_{15}\left(\frac{\partial u_2}{\partial x_2}\right)^2=0,
\end{aligned}
\end{equation}
with $u_1(x_1,x_2)$, $u_2(x_1,x_2)$ scalar functions, and $\kappa_i\left(u_{1},u_{2}\right)$ $(i=1,\ldots ,15)$ arbitrary smooth functions of the indicated arguments.

It can be easily ascertained that system (\ref{Oliv_sys}) admits the Lie point symmetries spanned by the operators
\begin{equation}
\begin{aligned}
&\Xi_1=\frac{\partial}{\partial x_1},\qquad \Xi_2=\frac{\partial}{\partial x_2},\\
&\Xi_3=\left(x_1-a u_1-b u_2 \right)\frac{\partial}{\partial x_1}+\left(x_2-b u_1- c u_2\right)\frac{\partial}{\partial x_2},
\end{aligned}
\end{equation}
$a$, $b$, $c$ being constants, provided that the conditions 
\begin{equation}
\label{Oliv_kappa}
\begin{aligned}
&\kappa_1-c^2 \kappa_{10}+ b c \kappa_{11}-a c \kappa_{12}-b^2 \kappa_{13}+a b \kappa_{14}-a^2 \kappa_{15}=0,\\
&\kappa_2-c^2 \kappa_{10}+ b c \kappa_{11}-a c \kappa_{12}-b^2 \kappa_{13}+a b \kappa_{14}-a^2 \kappa_{15}=0,\\
&\kappa_3+2(c^2 \kappa_{10}- b c \kappa_{11}+a c \kappa_{12}+b^2 \kappa_{13}-a b \kappa_{14}+a^2 \kappa_{15})=0,\\
&\kappa_4+2 c \kappa_{10}- b \kappa_{11}+a \kappa_{12}=0,\\
&\kappa_5-2 c \kappa_{10}+ b \kappa_{11}-a \kappa_{12}=0,\\
&\kappa_6+ c \kappa_{11}-2 b \kappa_{13}+a \kappa_{14}=0,\\
&\kappa_7- c \kappa_{11}+2 b \kappa_{13}-a \kappa_{14}=0,\\
&\kappa_8+ c \kappa_{12}- b \kappa_{14}+2 a \kappa_{15}=0,\\
&\kappa_9- c \kappa_{12}+ b \kappa_{14}-2 a \kappa_{15}=0
\end{aligned}
\end{equation}
are satisfied.
Since
\begin{equation}
\left[\Xi_1,\Xi_2\right]=0,\qquad \left[\Xi_1,\Xi_3\right]=\Xi_1,\qquad \left[\Xi_2,\Xi_3\right]=\Xi_2,
\end{equation}
applying the theorem, we introduce the new independent and dependent variables
\begin{equation}
\label{newvars_oli}
\begin{aligned}
&z_1=x_1-a u_1-b u_2, \quad && z_2=x_2-b u_1-c u_2,\\
&w_1  =u_1,   &&w_2=u_2,
\end{aligned}
\end{equation}
and the nonlinear system (\ref{Oliv_sys}) 
reduces to
\begin{equation}
\label{Oli_final}
\begin{aligned}
& \frac{\partial w_1}{\partial z_2}-\frac{\partial w_2}{\partial z_1}=0,\\
&\kappa_{10}\left(\frac{\partial w_1}{\partial z_1}\right)^2+\kappa_{11}\frac{\partial w_1}{\partial z_1}\frac{\partial w_1}{\partial z_2}+\kappa_{12}\frac{\partial w_1}{\partial z_1}\frac{\partial w_2}{\partial z_2}\\
&\quad+\kappa_{13}\left(\frac{\partial w_1}{\partial z_2}\right)^2
+\kappa_{14}\frac{\partial w_1}{\partial z_2}\frac{\partial w_2}{\partial z_2}+\kappa_{15}\left(\frac{\partial w_2}{\partial z_2}\right)^2=0,
\end{aligned}
\end{equation}
\emph{i.e.}, reads as an autonomous system polynomially homogeneous in the derivatives. 

We notice that system (\ref{Oli_final}), by specializing the functions $\kappa_{10},\ldots, \kappa_{15}$ as follows, 
\begin{equation}
\begin{aligned}
& \kappa_{10}=-\kappa(1+w_2^2)^2,\\
& \kappa_{11}=4\kappa w_1 w_2(1+w_2^2),\\
& \kappa_{12}=2((2-\kappa)(1+w_1^2+w_2^2)-\kappa w_1^2 w_2^2),\\
& \kappa_{13}=-4(1+w_1^2+w_2^2+\kappa w_1^2 w_2^2),\\
& \kappa_{14}=4\kappa w_1 w_2(1+w_1^2),\\ 
& \kappa_{15}=-\kappa(1+w_1^2)^2,\
\end{aligned}
\end{equation}
is equivalent to the second order partial differential equation
\begin{equation}
\label{mysurf}
\begin{aligned}
&\kappa\left(1+w_{z_2}^2\right)^2
w_{z_1z_1}^2-
4\kappa w_{z_1} w_{z_2}\left(1+w_{z_2}^2\right)w_{z_1z_1}
w_{z_1z_2}\\
&-2\left((2-\kappa)\left(1+w_{z_1}^2+w_{z_2}^2\right)-\kappa 
w_{z_1}^2w_{z_2}^2\right)w_{z_1z_1}w_{z_2z_2}\\
&+4\left(1+w_{z_1}^2+w_{z_2}^2+\kappa w_{z_1}^2w_{z_2}^2\right)
w_{z_1z_2}^2-4\kappa w_{z_1}w_{z_2}\left(1+w_{z_1}^2\right)
w_{z_1z_2}w_{z_2z_2}\\
&+\kappa\left(1+w_{z_1}^2\right)^2
w_{z_2z_2}^2=0,
\end{aligned}
\end{equation} 
where $\displaystyle w_{z_1}=\frac{\partial w}{\partial z_1}=w_1$, 
$\displaystyle w_{z_2}=\frac{\partial w}{\partial z_2}=w_2$, and  
$\kappa$ is an arbitrary function of $w_{z_1}$ and $w_{z_2}$.

Considering a smooth surface in $\mathbb{R}^3$ with the metric  
$ds^2=dz_1^2+ dz_2^2+dw^2$, and its Gaussian and mean curvature,
\begin{equation}
\begin{aligned}
&G=\frac{w_{z_1z_1}w_{z_2z_2}-w_{z_1z_2}^2}{(1+w_{z_1}^2+w_{z_2}^2)^{2}},\\
&H=\frac{1}{2}\frac{(1+w_{z_2}^2) w_{z_1z_1}-2 w_{z_1} w_{z_2} w_{z_1z_2}+(1+w_{z_1}^2) w_{z_2z_2}}{(1+w_{z_1}^2+w_{z_2}^2)^{3/2}},
\end{aligned}
\end{equation}
respectively, equation (\ref{mysurf}) can be written as
\begin{equation}
\label{surface}
G=\kappa H^2,
\end{equation}
whereupon it should be $\kappa(w_{z_1},w_{z_2})\le 1$.
In the limit case $\kappa\equiv 1$, Eq.~(\ref{surface}) characterizes a surface with all its points umbilic; it is known that a surface with all its point umbilic is a (open) domain of a plane or a sphere \cite{Pressley}.  
It is worth of being remarked that Eq.~(\ref{surface}) with $\kappa\equiv1$ is 
\emph{strongly Lie remarkable} \cite{MOVJMAA2007}, since it is the unique second 
order partial differential equation uniquely characterized by the conformal Lie algebra 
in $\mathbb{R}^3$ \cite{MOVTMP2007}.
\end{example}

\begin{remark}
Notice that the Example~\ref{mov} provides a system polynomial in the derivatives 
which is transformed into a system polynomially homogeneous of degree 2 in the derivatives. In 
Section~\ref{applications}, we will analyze various first order Monge--Amp\`ere  systems, and show that they can be transformed to quasilinear form.
\end{remark}

\section{Applications}
\label{applications}
In this section, we provide some examples of first order nonlinear systems polynomial in the  derivatives whose Lie symmetries satisfy the conditions of Theorem~\ref{th_1}, 
and prove that they can be transformed under suitable conditions to autonomous first order systems having homogeneous polynomial form; the systems that will be considered are of Monge--Amp\`ere type, and, remarkably, they are reduced to quasilinear (or linear) form.

In particular, we are concerned with the nonlinear first order systems of Monge--Amp\`ere equations for the unknowns $u_\alpha(x_i)$ $(\alpha=1,\ldots,m,\;i=1,\ldots,n)$. These systems have been characterized 
by  Boillat in 1997 \cite{Boillat.2} by looking for the nonlinear first order systems possessing, as the quasilinear systems, the property of the linearity of the Cauchy problem. 
These systems are also completely exceptional \cite{Lax,Boillat}, and are made by equations which are 
expressed as linear combinations (with coefficients depending at most on the independent and the dependent variables) of all minors extracted from the gradient matrix of $u_\alpha=u_\alpha(x_i)$.

Hereafter, to shorten the formulas, we denote with $u_{\alpha,i}$ the first order partial derivative of $u_\alpha(x_i)$ with respect to $x_i$, and with $w_{\alpha,i}$ the first order partial derivative of $w_\alpha(z_i)$ with respect to $z_i$; moreover, we denote with $f_{i;\alpha}$ the first order partial derivative of the function $f_i$ with respect to $u_\alpha$ (or $w_\alpha$). 
In the following we limit ourselves to consider the coefficients of the Monge--Amp\`ere systems at most functions
of the field variables.

\subsection{Case $m=n=2$}
Let us consider the nonlinear first order system of  Monge--Amp\`ere made by the equations
\begin{equation}
\label{monge22}
\kappa^i_0\left(u_{1,1}u_{2,2}-u_{1,2}u_{2,1}\right)+\kappa^i_1u_{1,1}+\kappa^i_2u_{1,2}+\kappa^i_3u_{2,1}+\kappa^i_4u_{2,2}+\kappa^i_5=0 
\end{equation}
$(i=1,2)$, with $u_1(x_1,x_2),u_2(x_1,x_2)$ scalar functions, and $\kappa^i_j\left(u_{1},u_{2}\right)$ $(i=1,2; \; j=0,\ldots ,5)$ arbitrary smooth functions of the indicated arguments.

The substitutions
\begin{equation}
u_1\rightarrow u_1+\alpha_{11}x_1+\alpha_{12}x_2, \qquad u_2\rightarrow u_2+\alpha_{21}x_1+\alpha_{22}x_2,  
\end{equation}
where $\alpha_{ij}$ are arbitrary constants, produce a system with $\kappa^i_5=0$ $(i=1,2)$ provided that
\begin{equation}
\label{condkappa5}
\kappa^i_0(\alpha_{11}\alpha_{22}-\alpha_{12}\alpha_{21})+\kappa^i_1\alpha_{11}+\kappa^i_2\alpha_{12}+\kappa^i_3\alpha_{21}+\kappa^i_4\alpha_{22}+\kappa^i_5=0
\end{equation}
$(i=1,2)$. Conditions (\ref{condkappa5})  provide two constraints on the functional form of the coefficients so that not all systems can be written in a form where $\kappa^i_5=0$; however, if the coefficients $\kappa^i_j$ are constant,
due to the arbitrariness of the constants $\alpha_{ij}$, then (\ref{condkappa5}) can always be satisfied whatever the values of the coefficients are.

It is easily recognized that system (\ref{monge22}), now taken with $\kappa^i_5=0$, admits the Lie point symmetries spanned by the operators
\begin{equation}
\Xi_1=\frac{\partial}{\partial x_1},\qquad
\Xi_2=\frac{\partial}{\partial x_2},\qquad
\Xi_3=\left(x_1-f_{1}\right)\frac{\partial}{\partial x_1}+\left(x_2-f_{2}\right)\frac{\partial}{\partial x_2},
\end{equation}
where $f_i(u_1,u_2)$ $(i=1,2)$ are arbitrary smooth functions of their arguments, provided that 
\begin{equation}
\label{condkappa0}
\kappa^i_0+\kappa^i_1f_{2;2}-\kappa^i_2f_{1;2}-\kappa^i_3f_{2;1}+\kappa^i_4f_{1;1}=0, \qquad i=1,2.
\end{equation}
The constraints (\ref{condkappa0}), once we assign the 10 functions $k^i_j(u_1,u_2)$ ($i=1,2$, $j=0,\ldots,4$),
are the differential equations providing us the functional form of $f_1(u_1,u_2)$ and $f_2(u_1,u_2)$.

In the case where all the coefficients $\kappa^i_j$ are constant, then the functions $f_1$ and $f_2$ are forced to be linear, \emph{i.e.},
\begin{equation}
f_1=\beta_{11}u_1+\beta_{12}u_2, \qquad f_2=\beta_{21}u_1+\beta_{22}u_2,
\end{equation} 
$\beta_{ij}$ being constants whose value is determined by the coefficients $\kappa^i_j$.

Since
\begin{equation}
\left[\Xi_1,\Xi_2\right]=0,\qquad \left[\Xi_1,\Xi_3\right]=\Xi_1,\qquad \left[\Xi_2,\Xi_3\right]=\Xi_2,
\end{equation}
we introduce the new variables
\begin{equation}
\label{newvars11}
z_1=x_1-f_{1},\qquad z_2=x_2-f_{2},\qquad w_1=u_1, \qquad w_2=u_2,
\end{equation}
and the generators of the point symmetries write as
\begin{equation}
\Xi_1=\frac{\partial}{\partial z_1}, \qquad \Xi_2=\frac{\partial}{\partial z_2}, \qquad
\Xi_3=z_1\frac{\partial}{\partial z_1}+z_2\frac{\partial}{\partial z_2}.
\end{equation}
In terms of the new variables (\ref{newvars11}),  the nonlinear system (\ref{monge22}) 
becomes
\begin{equation}
\label{monge22final}
\kappa^i_1w_{1,1}+\kappa^i_2w_{1,2}+\kappa^i_3w_{2,1}+\kappa^i_4w_{2,2}=0,
\end{equation}
\emph{i.e.}, reads as an autonomous and homogeneous quasilinear system. This system is linear if all
the coefficients $\kappa^i_j$ are constant; nevertheless, since it is a $2\times 2$ homogeneous and autonomous
quasilinear system, it can be written in linear form by means of the hodograph transformation also when
the coefficients $\kappa^i_j$ depend on $u_1$ and $u_2$.

In conclusion, all Monge--Amp\`ere systems with $m=n=2$ can be reduced to a linear system when the coefficients $\kappa^i_j$ are constant; on the contrary, when the coefficients depend upon $u_1$ and $u_2$ the reduction to the linear form is possible provided that the constraints (\ref{condkappa5}) are satisfied.

\subsection{Case $m=2,\,n=3$}
By considering the gradient matrix of $u_\alpha(x_i)$ $(\alpha=1,2,\;i=1,\ldots,3)$
\begin{equation}
H=\left(\begin{array}{ccc}
u_{1,1} & u_{1,2} & u_{1,3} \\ 
u_{2,1} & u_{2,2} & u_{2,3}
\end{array} \right)
\end{equation}
and its extracted minors
\begin{equation}
H^1=\left|\begin{array}{cc}
u_{1,1} & u_{1,2}  \\ 
u_{2,1} & u_{2,2}
\end{array} \right|,\qquad
H^2=\left|\begin{array}{cc}
u_{1,1} & u_{1,3}  \\ 
u_{2,1} & u_{2,3}
\end{array} \right|,\qquad
H^3=\left|\begin{array}{cc}
u_{1,2} & u_{1,3}  \\ 
u_{2,2} & u_{2,3}
\end{array} \right|,
\end{equation}
the nonlinear first order system of  Monge--Amp\`ere is made by equations like
\begin{equation}
\label{monge23}
\begin{aligned}
&\kappa^i_1H^1+\kappa^i_2 H^2+\kappa^i_3 H^3\\
&\quad+\kappa^i_4u_{1,1}+\kappa^i_5u_{1,2}+\kappa^i_6u_{1,3}+\kappa^i_7u_{2,1}+\kappa^i_8u_{2,2}+\kappa^i_9u_{2,3}+\kappa^i_{10}=0,
\end{aligned}
\end{equation}
$\kappa^i_j\left(u_{\alpha}\right)$ $(i=1,2,\; j=1,\ldots ,10)$ arbitrary smooth functions of the indicated arguments.

The substitutions
\begin{equation}
u_1\rightarrow u_1+\alpha_{11}x_1+\alpha_{12}x_2+\alpha_{13}x_3, \qquad u_2\rightarrow u_2+\alpha_{21}x_1+\alpha_{22}x_2+\alpha_{23}x_3,  
\end{equation}
where $\alpha_{ij}$ are arbitrary constants, produce a system with $\kappa^i_{10}=0$ $(i=1,2)$
provided that 
\begin{equation}
\label{condkappa10}
\begin{aligned}
&\kappa^i_1(\alpha_{11}\alpha_{22}-\alpha_{12}\alpha_{21})+
\kappa^i_2(\alpha_{11}\alpha_{23}-\alpha_{13}\alpha_{21})+
\kappa^i_3(\alpha_{12}\alpha_{23}-\alpha_{13}\alpha_{22})\\
&\quad+\kappa^i_4\alpha_{11}+\kappa^i_5\alpha_{12}+\kappa^i_6\alpha_{13}+\kappa^i_7\alpha_{21}+\kappa^i_8\alpha_{22}+\kappa^i_9\alpha_{23}+\kappa^i_{10}=0
\end{aligned}
\end{equation}
$(i=1,2)$. Actually, conditions (\ref{condkappa10}) can always be satisfied when the coefficients $\kappa^i_j$
are constant because of the arbitrariness of the constants $\alpha_{ij}$.
 
This nonlinear system (\ref{monge23}), now taken with $\kappa^i_{10}=0$, admits the Lie point symmetries spanned by the operators
\begin{equation}
\label{sym23}
\begin{aligned}
&\Xi_1=\frac{\partial}{\partial x_1},\qquad \Xi_2=\frac{\partial}{\partial x_2},\qquad \Xi_3=\frac{\partial}{\partial x_3},\\
&\Xi_4=\left(x_1-f_{1}\right)\frac{\partial}{\partial x_1}+\left(x_2-f_{2}\right)\frac{\partial}{\partial x_2}+\left(x_3-f_{3}\right)\frac{\partial}{\partial x_3},
\end{aligned}
\end{equation}
where $f_i(u_1,u_2)$ $(i=1,\ldots,3)$ are arbitrary smooth functions of their arguments, provided that 
\begin{equation}
\label{condkappa123}
\begin{aligned}
&\kappa^i_1+\kappa^i_4f_{2;2}-\kappa^i_5f_{1;2}-\kappa^i_7f_{2;1}+\kappa^i_8f_{1;1}=0,\\
&\kappa^i_2+\kappa^i_4f_{3;2}-\kappa^i_6f_{1;2}-\kappa^i_7f_{3;1}+\kappa^i_9f_{1;1}=0,\\
&\kappa^i_3+\kappa^i_5f_{3;2}-\kappa^i_6f_{2;2}-\kappa^i_8f_{3;1}+\kappa^i_9f_{2;1}=0.
\end{aligned}
\end{equation}
The six conditions (\ref{condkappa123}) cannot be fulfilled for an arbitrary choice of the coefficients $\kappa^i_j$. In the simplest case, where the coefficients $\kappa^i_j$ are constant, they can be always satisfied and the functions $f_i$ must be linear:
\begin{equation}
f_1=\beta_{11} u_1+\beta_{12}u_2,\quad f_2=\beta_{21} u_1+\beta_{22}u_2,
\quad f_3=\beta_{31} u_1+\beta_{32}u_2,
\end{equation}
$\beta_{ij}$ being arbitrary constants.

The symmetries (\ref{sym23}) generate a 4--dimensional solvable Lie algebra,
\begin{equation}
\left[\Xi_i,\Xi_j\right]=0,\qquad \left[\Xi_i,\Xi_4\right]=\Xi_i,
\qquad (i,j=1,2,3),
\end{equation}
whereupon we may introduce the new variables
\begin{equation}
\label{newvars23}
\begin{array}{lll}
z_1=x_1-f_{1},\quad &z_2=x_2-f_{2},\quad & z_3= x_3-f_{3},\\ 
w_1=u_1,\quad  &w_2=u_2,
\end{array}
\end{equation}
and the generators of the point symmetries write as
\begin{equation}
\Xi_1=\frac{\partial}{\partial z_1}, \quad \Xi_2=\frac{\partial}{\partial z_2}, \quad
\Xi_3=\frac{\partial}{\partial z_3}, \quad
\Xi_4=z_1\frac{\partial}{\partial z_1}+z_2\frac{\partial}{\partial z_2}
+z_3\frac{\partial}{\partial z_3}. 
\end{equation}

In terms of the new variables (\ref{newvars23}), the nonlinear system (\ref{monge23}) reduces to
\begin{equation}
\kappa^i_4w_{1,1}+\kappa^i_5w_{1,2}+\kappa^i_6w_{1,3}+\kappa^i_7w_{2,1}+\kappa^i_8w_{2,2}+\kappa^i_9w_{2,3}=0,
\end{equation}
\emph{i.e.}, reads as an autonomous and homogeneous quasilinear system.

\subsection{Case $m=3,\,n=2$}
By considering the gradient matrix of $u_\alpha(x_i)$ $(\alpha=1,\ldots,3,\;i=1,2)$
\begin{equation}
H=\left(\begin{array}{cc}
u_{1,1} & u_{1,2}  \\ 
u_{2,1} & u_{2,2} \\
u_{3,1} & u_{3,2} \\
\end{array} \right)
\end{equation}
and its extracted minors
\begin{equation}
H^1=\left|\begin{array}{cc}
u_{1,1} & u_{1,2}  \\ 
u_{2,1} & u_{2,2}
\end{array} \right|,\qquad
H^2=\left|\begin{array}{cc}
u_{1,1} & u_{1,2}  \\ 
u_{3,1} & u_{3,2}
\end{array} \right|,\qquad
H^3=\left|\begin{array}{cc}
u_{2,1} & u_{2,2}  \\ 
u_{3,1} & u_{3,2}
\end{array} \right|,
\end{equation}
the nonlinear first order system of  Monge--Amp\`ere is made by equations like
\begin{equation}
\begin{aligned}\label{monge32}
&\kappa^i_1H^1+\kappa^i_2 H^2+\kappa^i_3 H^3\\
&\quad+\kappa^i_4u_{1,1}+\kappa^i_5u_{1,2}+\kappa^i_6u_{2,1}+\kappa^i_7u_{2,2}+\kappa^i_8u_{3,1}+\kappa^i_9u_{3,2}+\kappa^i_{10}=0,
\end{aligned}
\end{equation}
$\kappa^i_j\left(u_{\alpha}\right)$ $(i=1,\ldots,3,\; j=1,\ldots ,10)$ arbitrary smooth functions of the indicated arguments.

The substitutions
\begin{equation}
\begin{aligned}
&u_1\rightarrow u_1+\alpha_{11}x_1+\alpha_{12}x_2,\\
&u_2\rightarrow u_2+\alpha_{21}x_1+\alpha_{22}x_2,\\  
&u_3\rightarrow u_3+\alpha_{31}x_1+\alpha_{32}x_2,\\
\end{aligned}
\end{equation}
where $\alpha_{ij}$ are arbitrary constants, produce a system with $\kappa^i_{10}=0$ $(i=1,\ldots,3)$
provided that 
\begin{equation}
\label{condkappa10_2}
\begin{aligned}
&\kappa^i_1(\alpha_{11}\alpha_{22}-\alpha_{12}\alpha_{21})+
\kappa^i_2(\alpha_{11}\alpha_{32}-\alpha_{12}\alpha_{31})+
\kappa^i_3(\alpha_{21}\alpha_{32}-\alpha_{22}\alpha_{31})\\
&\quad+\kappa^i_4\alpha_{11}+\kappa^i_5\alpha_{12}+\kappa^i_6\alpha_{21}+\kappa^i_7\alpha_{22}+\kappa^i_8\alpha_{31}+\kappa^i_9\alpha_{32}+\kappa^i_{10}=0
\end{aligned}
\end{equation}
$(i=1,\ldots,3)$. Also in this case, conditions (\ref{condkappa10_2}) can always be satisfied when the coefficients 
$\kappa^i_j$ are constant because of the arbitrariness of the constants $\alpha_{ij}$.

This nonlinear system (\ref{monge32}), now taken with $\kappa^i_{10}=0$, admits the Lie point symmetries spanned by the operators
\begin{equation}
\label{sym32}
\Xi_1=\frac{\partial}{\partial x_1},\qquad \Xi_2=\frac{\partial}{\partial x_2},\qquad \Xi_3=\left(x_1-f_{1}\right)\frac{\partial}{\partial x_1}+\left(x_2-f_{2}\right)\frac{\partial}{\partial x_2},
\end{equation}
where $f_i(u_1,u_2,u_3)$ $(i=1,2)$ are arbitrary smooth functions of their arguments, provided that 
\begin{equation}
\label{condkappa123_2}
\begin{aligned}
&\kappa^i_1+\kappa^i_4f_{2;2}-\kappa^i_5f_{1;2}-\kappa^i_6f_{2;1}+\kappa^i_7f_{1;1}=0,\\
&\kappa^i_2+\kappa^i_4f_{2;3}-\kappa^i_5f_{1;3}-\kappa^i_8f_{2;1}+\kappa^i_9f_{1;1}=0,\\
&\kappa^i_3+\kappa^i_6f_{2;3}-\kappa^i_7f_{1;3}-\kappa^i_8f_{2;2}+\kappa^i_9f_{1;2}=0.
\end{aligned}
\end{equation}
The six conditions (\ref{condkappa123_2}) cannot be fulfilled for an arbitrary choice of the coefficients $\kappa^i_j$. In the simplest case, where the coefficients $\kappa^i_j$ are constant, they can be always satisfied and the functions $f_i$ must be linear:
\begin{equation}
f_1=\beta_{11} u_1+\beta_{12}u_2+\beta_{13}u_3,\qquad f_2=\beta_{21} u_1+\beta_{22}u_2+\beta_{23}u_3,
\end{equation}
$\beta_{ij}$ being arbitrary constants.

The symmetries (\ref{sym32}) generate a 3--dimensional solvable Lie algebra,
\begin{equation}
\left[\Xi_i,\Xi_j\right]=0,\qquad \left[\Xi_i,\Xi_3\right]=\Xi_i,
\qquad (i,j=1,2),
\end{equation}
whereupon we may introduce the new variables
\begin{equation}
\label{newvars32}
\begin{array}{lll}
z_1=x_1-f_{1},\quad &z_2=x_2-f_{2},\\ 
w_1=u_1,\quad  &w_2=u_2,\quad  &w_3=u_3,
\end{array}
\end{equation}
and the generators of the point symmetries write as
\begin{equation}
\Xi_1=\frac{\partial}{\partial z_1}, \qquad \Xi_2=\frac{\partial}{\partial z_2}, \qquad
\Xi_4=z_1\frac{\partial}{\partial z_1}+z_2\frac{\partial}{\partial z_2}. 
\end{equation}
In terms of the new variables (\ref{newvars32}), the nonlinear system (\ref{monge32}) reduces to
\begin{equation}
\kappa^i_4w_{1,1}+\kappa^i_5w_{1,2}+\kappa^i_6w_{2,1}+\kappa^i_7w_{2,2}+\kappa^i_8w_{3,1}+\kappa^i_9w_{3,2}=0,
\end{equation}
\emph{i.e.}, reads as an autonomous and homogeneous quasilinear system.

\subsection{Case $m=n=3$}
By considering the gradient matrix of $u_\alpha(x_i)$ $(i=1,\ldots,3,\, \alpha=1,\ldots,3)$
\begin{equation}
H=\left(\begin{array}{ccc}
u_{1,1} & u_{1,2} & u_{1,3} \\ 
u_{2,1} & u_{2,2} & u_{2,3} \\
u_{3,1} & u_{3,2} & u_{3,3} \\
\end{array} \right)
\end{equation}
and its extracted minors of order 2
\begin{equation*}
\begin{aligned}
&H^1=\left|\begin{array}{cc}
u_{2,2} & u_{2,3}  \\ 
u_{3,2} & u_{3,3}
\end{array} \right|,\qquad
H^2=\left|\begin{array}{cc}
u_{2,1} & u_{2,3}  \\ 
u_{3,1} & u_{3,3}
\end{array} \right|,\qquad
H^3=\left|\begin{array}{cc}
u_{2,1} & u_{2,2}  \\ 
u_{3,1} & u_{3,2}
\end{array} \right|,\\
&H^4=\left|\begin{array}{cc}
u_{1,2} & u_{1,3}  \\ 
u_{3,2} & u_{3,3}
\end{array} \right|,\qquad
H^5=\left|\begin{array}{cc}
u_{1,1} & u_{1,3}  \\ 
u_{3,1} & u_{3,3}
\end{array} \right|,\qquad
H^6=\left|\begin{array}{cc}
u_{1,1} & u_{1,2}  \\ 
u_{3,1} & u_{3,2}
\end{array} \right|,\\
&H^7=\left|\begin{array}{cc}
u_{1,2} & u_{1,3}  \\ 
u_{2,2} & u_{2,3}
\end{array} \right|,\qquad
H^8=\left|\begin{array}{cc}
u_{1,1} & u_{1,3}  \\ 
u_{2,1} & u_{2,3}
\end{array} \right|,\qquad
H^9=\left|\begin{array}{cc}
u_{1,1} & u_{1,2}  \\ 
u_{2,1} & u_{2,2}
\end{array} \right|,\\
\end{aligned}
\end{equation*}
the nonlinear first order system of  Monge--Amp\`ere results composed by equations like
\begin{equation}
\begin{aligned}\label{monge33}
&\kappa^i_0 \det(H)+\kappa^i_1H^1+\kappa^i_2 H^2+\kappa^i_3 H^3+\kappa^i_4H^4+\kappa^i_5 H^5+\kappa^i_6 H^6\\
&\quad+\kappa^i_7H^7+\kappa^i_8 H^8+\kappa^i_9 H^9+\kappa^i_{10}u_{1,1}+\kappa^i_{11}u_{1,2}+\kappa^i_{12}u_{1,3}
+\kappa^i_{13}u_{2,1}\\
&\quad+\kappa^i_{14}u_{2,2}+\kappa^i_{15}u_{2,3}
+\kappa^i_{16}u_{3,1}+\kappa^i_{17}u_{3,2}+\kappa^i_{18}u_{3,3}+\kappa^i_{19}=0,
\end{aligned}
\end{equation}
$\kappa^i_j\left(u_{\alpha}\right)$ $(i=1,\ldots,3,\; j=0,\ldots ,19,\;\alpha=1,\ldots,3)$ arbitrary smooth functions of the indicated arguments.

Also in this case, the substitutions
\begin{equation}
\begin{aligned}
&u_1\rightarrow u_1+\alpha_{11}x_1+\alpha_{12}x_2+\alpha_{13}x_3,\\
&u_2\rightarrow u_2+\alpha_{21}x_1+\alpha_{22}x_2+\alpha_{23}x_3,\\  
&u_3\rightarrow u_3+\alpha_{31}x_1+\alpha_{32}x_2+\alpha_{33}x_3,
\end{aligned}
\end{equation}
where $\alpha_{ij}$ are arbitrary constants, allow us to obtain a system with $\kappa^i_{19}=0$ 
provided that $u_{i,j}=\alpha_{ij}$ is a solution of equations (\ref{monge33}). This requirement  implies some constraints on the coefficients $\kappa^i_j$ in the general case; on the contrary, no limitation to the values of the coefficients exists if they are assumed to be constant.

The system (\ref{monge33}), with $\kappa^i_{19}=0$, admits the Lie point symmetries spanned by the operators
\begin{equation}
\label{sym33}
\begin{aligned}
&\Xi_1=\frac{\partial}{\partial x_1},\qquad \Xi_2=\frac{\partial}{\partial x_2},\qquad \Xi_3=\frac{\partial}{\partial x_3},\\
&\Xi_4=\left(x_1-f_{1}\right)\frac{\partial}{\partial x_1}+\left(x_2-f_{2}\right)\frac{\partial}{\partial x_2}+\left(x_3-f_{3}\right)\frac{\partial}{\partial x_3},
\end{aligned}
\end{equation}
where $f_i(u_1,u_2,u_3)$ $(i=1,\ldots,3)$ are arbitrary smooth functions of their arguments, provided that 
\begin{equation}
\label{cond33}
\begin{aligned}
&\kappa^i_{0}-(f_{2;2}f_{3;3}-f_{2;3}f_{3;2})\kappa^i_{10}-(f_{1;3}f_{3;2}-f_{1;2}f_{3;3})\kappa^i_{11}\\
&\quad-(f_{1;2}f_{2;3}-f_{1;3}f_{2;2})\kappa^i_{12}-(f_{2;3}f_{3;1}-f_{2;1}f_{3;3})\kappa^i_{13}\\
&\quad-(f_{1;1}f_{3;3}-f_{1;3}f_{3;1})\kappa^i_{14}-(f_{1;3}f_{2;1}-f_{1;1}f_{2;3})\kappa^i_{15}\\
&\quad-(f_{2;1}f_{3;2}-f_{2;2}f_{3;1})\kappa^i_{16}-(f_{1;2}f_{3;1}-f_{1;1}f_{3;2})\kappa^i_{17}\\
&\quad-(f_{1;1}f_{2;2}-f_{1;2}f_{2;1})\kappa^i_{18}=0,\\
&\kappa^i_{1}+\kappa^i_{14}f_{3;3}-\kappa^i_{15}f_{2;3}-\kappa^i_{17}f_{3;2}+\kappa^i_{18}f_{2;2}=0,\\
&\kappa^i_{2}+\kappa^i_{13}f_{3;3}-\kappa^i_{15}f_{1;3}-\kappa^i_{16}f_{3;2}+\kappa^i_{18}f_{1;2}=0,\\
&\kappa^i_{3}+\kappa^i_{13}f_{2;3}-\kappa^i_{14}f_{1;3}-\kappa^i_{16}f_{2;2}+\kappa^i_{17}f_{1;2}=0,\\
&\kappa^i_{4}+\kappa^i_{11}f_{3;3}-\kappa^i_{12}f_{2;3}-\kappa^i_{17}f_{3;1}+\kappa^i_{18}f_{2;1}=0,\\
&\kappa^i_{5}+\kappa^i_{10}f_{3;3}-\kappa^i_{12}f_{1;3}-\kappa^i_{16}f_{3;1}+\kappa^i_{18}f_{1;1}=0,\\
&\kappa^i_{6}+\kappa^i_{10}f_{2;3}-\kappa^i_{11}f_{1;3}-\kappa^i_{16}f_{2;1}+\kappa^i_{17}f_{1;1}=0,\\
&\kappa^i_{7}+\kappa^i_{11}f_{3;2}-\kappa^i_{12}f_{2;2}-\kappa^i_{14}f_{3;1}+\kappa^i_{15}f_{2;1}=0,\\
&\kappa^i_{8}+\kappa^i_{10}f_{3;2}-\kappa^i_{12}f_{1;2}-\kappa^i_{13}f_{3;1}+\kappa^i_{15}f_{1;1}=0,\\
&\kappa^i_{9}+\kappa^i_{10}f_{2;2}-\kappa^i_{11}f_{1;2}-\kappa^i_{13}f_{2;1}+\kappa^i_{14}f_{1;1}=0.\\
\end{aligned}
\end{equation}
The vector fields (\ref{sym33}) span a 4--dimensional solvable Lie algebra,
\begin{equation}
\left[\Xi_i,\Xi_j\right]=0,\qquad \left[\Xi_i,\Xi_4\right]=\Xi_i,
\qquad (i,j=1,\ldots,3),
\end{equation}
whereupon we may introduce the new variables
\begin{equation}
\label{newvars33}
\begin{array}{lll}
z_1=x_1-f_{1},\quad &z_2=x_2-f_{2},\quad &z_3=x_3-f_{3},\\ 
w_1=u_1,\quad  &w_2=u_2,\quad  &w_3=u_3,
\end{array}
\end{equation}
and the generators of the point symmetries write as
\begin{equation}
\begin{aligned}
&\Xi_1=\frac{\partial}{\partial z_1}, \qquad \Xi_2=\frac{\partial}{\partial z_2},\qquad \Xi_3=\frac{\partial}{\partial z_3},\\
&\Xi_4=z_1\frac{\partial}{\partial z_1}+z_2\frac{\partial}{\partial z_2}+z_3\frac{\partial}{\partial z_3}.
\end{aligned} 
\end{equation}

In terms of the new variables (\ref{newvars33}), Eqs.~(\ref{monge33}) write as
\begin{equation}
\begin{aligned}
&\kappa^i_{10}w_{1,1}+\kappa^i_{11}w_{1,2}+\kappa^i_{12}w_{1,3}+\kappa^i_{13}w_{2,1}+\kappa^i_{14}w_{2,2}+\kappa^i_{15}w_{2,3}\\
&+\kappa^i_{16}w_{3,1}+\kappa^i_{17}w_{3,2}+\kappa^i_{18}w_{3,3}=0,
\end{aligned}
\end{equation}
\emph{i.e.}, they are in autonomous and homogeneous quasilinear (linear, if the coefficients are constant) form.

Conditions (\ref{cond33}) play severe restrictions to the expression of the coefficients $\kappa^i_j$. 
When these coefficients are assumed to be constant, we are forced to take
\begin{equation}
\begin{aligned}
&f_1=\beta_{11} u_1+\beta_{12}u_2+\beta_{13}u_3,\\
&f_2=\beta_{21} u_1+\beta_{22}u_2+\beta_{23}u_3,\\
&f_3=\beta_{31} u_1+\beta_{32}u_2+\beta_{33}u_3,
\end{aligned}
\end{equation}
where $\beta_{ij}$ are arbitrary constants; also in such simple case, the reduction to linear form is not always possible due to (\ref{cond33}).

\subsection{Case $m$ and $n$ arbitrary}
It is easily recognized that,  a general Monge--Amp\`ere system with $m$ dependent variables and $n$ independent variables, provided that some suitable conditions on the coefficients  (at most depending on the
field variables) are satisfied, is invariant with respect to the Lie groups generated by the vector fields
\begin{equation}
\Xi_i=\frac{\partial }{\partial x_i}\; (i=1,\ldots,n), \qquad \Xi_{n+1}=\sum_{i=1}^n \left(x_i-f_i(u_\alpha)\right)\frac{\partial }{\partial x_i},
\end{equation}
where $f_i(u_\alpha)$ are smooth functions of ($u_1,\ldots,u_m)$ which have to be linear in their arguments when the coefficients of the Monge--Amp\`ere system are constant.

As one expects, for $m>3$ or $n>3$, we have a situation similar to the case $m=n=3$, \emph{i.e.}, even in the case of constant coefficients, not all Monge--Amp\`ere systems can be reduced to (quasi)linear form. 

\section{Conclusions}
\label{conclusions}
In this paper, we proved a theorem giving necessary and sufficient conditions for transforming a nonlinear first 
order system of partial differential equations involving the derivatives in polynomial form to an equivalent
autonomous system polynomially homogeneous in the derivatives. The theorem is based on the Lie point symmetries admitted by the nonlinear system, and the proof is constructive, in the sense that it leads to the algorithmic
construction of the invertible mapping performing the task.

The theorem is applied to a class of first order nonlinear systems belonging to the family of Monge--Amp\`ere systems that have been characterized by Boillat in 1997.  These systems share with the quasilinear systems the property of the linearity of the Cauchy problem. 
These systems are also completely exceptional \cite{Lax,Boillat}, and are made by equations which are 
expressed as linear combinations (with coefficients depending at most on the 
independent and the dependent variables) of all minors extracted from the gradient 
matrix of $u_\alpha=u_\alpha(x_i)$ $(\alpha=1,\ldots,m,\;i=1,\ldots,n)$. We considered 
explicitly either the case of constant coefficients or the case of coefficients depending 
on the field variables, for $m=2,3$ and $n=2,3$. If $m=2$ and $n=2,3$, or $n=2$ and 
$m=2,3$, and the coefficients  are assumed to be constant, we proved that the 
Monge--Amp\`ere systems can always be transformed to linear form.

Nevertheless, for arbitrary $m$ and $n$, Monge--Amp\`ere systems, provided that the 
coefficients entering their
equations satisfy some constraints, can be mapped to  first order quasilinear 
autonomous and homogeneous systems. This, in some sense, casts new light on the 
fact, underlined by Boillat \cite{Boillat.2}, that 
Monge--Amp\`ere systems, because of the linearity of the Cauchy problem, are the 
closest to quasilinear systems, which are Monge systems. 

Moreover, an example of a first order system polynomial in the derivatives that can be 
reduced to a system polynomially homogeneous in the derivatives (equivalent to a 
second order partial differential equation for a surface in $\mathbb{R}^3$ such that its 
Gaussian curvature is proportional to the square of its mean curvature), is provided.

\section*{Acknowledgments}
The authors thank the Referee for the helpful comments leading to improve the quality of the paper. Work supported by G.N.F.M. of the Istituto Nazionale di Alta Matematica.

\end{document}